%% file: EMLS-SS.tex
\documentclass[12pt]{article} 
\usepackage[sectionbib]{natbib}
\usepackage{array,epsfig,fancyheadings,rotating}
\usepackage{bbm}
\usepackage{mathtools}
\usepackage{epsfig,epstopdf}
\usepackage[colorlinks, citecolor=blue]{hyperref}  
\usepackage{sectsty, secdot}
\usepackage{breakcites}
\sectionfont{\fontsize{12}{14pt plus.8pt minus .6pt}\selectfont}
\renewcommand{\theequation}{\thesection\arabic{equation}}
\subsectionfont{\fontsize{12}{14pt plus.8pt minus .6pt}\selectfont}

\usepackage{amsmath}
\usepackage{amssymb}
\usepackage{amsfonts}
\usepackage{multirow}
\usepackage{amsthm}

\newcommand{\blue}{\color{blue}}

\input{symbols}

\begin{document}


\renewcommand{\baselinestretch}{1.6}

\markright{ \hbox{\footnotesize\rm 
}\hfill\\[-13pt]
\hbox{\footnotesize\rm
}\hfill }

\markboth{\hfill{\footnotesize\rm JIAHUA CHEN, PENGFEI LI AND GUANFU LIU} \hfill}
{\hfill {\footnotesize\rm HOMOGENEITY TESTING UNDER FINITE LOCATION-SCALE MIXTURES} \hfill}

\renewcommand{\thefootnote}{}
$\ $\par


\fontsize{12}{14pt plus.8pt minus .6pt}\selectfont \vspace{0.8pc}
\centerline{\large\bf Homogeneity testing under finite location-scale mixtures}
\vspace{.4cm} 
\centerline{Jiahua Chen$^{1}$, Pengfei Li$^{2}$, and Guanfu Liu$^{3}$} \vspace{.4cm} 
\begin{center}
{\it
$^1$Research Institute of
Big Data, Yunnan University, Kunming, Yunnan 650091, China\\
Department of Statistics, University of British Columbia, Vancouver, BC, Canada V6T 1Z2\\
E-mail: jhchen@stat.ubc.ca}
\end{center}

\begin{center}
{\it $^2$Department of Statistics and Actuarial Science, University of Waterloo, Waterloo, ON, Canada N2L 3G1\\
E-mail: pengfei.li@uwaterloo.ca} 
\end{center}

\begin{center}
{\it $^3$School of Statistics and Information, Shanghai University of International Business and Economics, Shanghai 201620, China\\
E-mail: liuguanfu07@163.com} 
\end{center}
\fontsize{9}{11.5pt plus.8pt minus.6pt}\selectfont


\begin{quotation}
\noindent {\it Abstract:}
The testing problem for the order of finite mixture models has
a long history and remains an active research topic.
Since \cite{Ghosh1985} revealed the hard-to-manage asymptotic
properties of the likelihood ratio test, there has been marked progress.
The most successful attempts include the modified likelihood
ratio test and the EM-test, which lead to neat solutions
for finite mixtures of univariate normal distributions, finite mixtures
of single-parameter distributions, and several mixture-like models.
{The problem remains challenging, and there is still no
generic solution for location-scale mixtures.}
In this paper, we provide an
EM-test solution for homogeneity for finite mixtures
of location-scale family distributions.
This EM-test has nonstandard limiting distributions, but we 
are able to find the critical values numerically.
We use computer experiments to obtain appropriate
values for the tuning parameters.
A simulation study shows that the fine-tuned EM-test has
close to nominal type I errors and very good power properties.
Two application examples are included to
demonstrate the performance of the EM-test.

\vspace{9pt}

\par
\end{quotation}\par

\def\thefigure{\arabic{figure}}
\def\thetable{\arabic{table}}

\renewcommand{\theequation}{\thesection.\arabic{equation}}

\fontsize{12}{14pt plus.8pt minus .6pt}\selectfont

\setcounter{equation}{0} 
\section{Introduction}

Let $\{f(x; \btheta): \btheta \in \Theta\}$ be a parametric distribution family.
A finite mixture model expands this family to include all convex combinations:

\[
f(x; G) = \sum_{j=1}^m \alpha_j f(x; \btheta_j),
\]
with the mixing distribution {$G(\btheta) $} given by
\[
G(\btheta) = \sum_{j=1}^m \alpha_j \ind (\btheta_j \leq \btheta).
\]
Here $\ind(\cdot)$ stands for the indicator function.
When $\btheta$ is a vector, the inequality is interpreted component-wise.
We may also write $G = \sum_{j=1}^m \alpha_j \{ \btheta_j\}$
and {regard it as the set of  all the parameters involved:
$\{(\alpha_j, \btheta_j): j=1, 2, \ldots, m\}$.
The subpopulation parameter space $\Theta$ is generally a
subset of an Euclidean space $\cR^d$ of dimension $d$.

\lhead[\footnotesize\thepage\fancyplain{}\leftmark]{}\rhead[]{\fancyplain{}\rightmark\footnotesize\thepage} 

In this paper, we consider the case where  $\btheta = (\mu, \sigma)^\tau$
and there exists a probability density function on $\cR$ with respect
to the Lebesgue measure $f_0(x)$ such that
\[
f(x; \btheta)
= \frac{1}{\sigma} f_0  \left (\frac{x-\mu}{\sigma} \right ).
\]
The parameter space for $\btheta$ is $\Theta = \cR \times \cR^+$,
with $\cR^+$ being all positive real numbers.
In other words, the subpopulation distributions are
members of a location-scale distribution family.
Location-scale mixtures are widely used in applications.
\cite{Naya2006} and \cite{Salimans2017}
applied mixtures of logistic distributions to
thermogravimetric analysis and imaging data, respectively.
{Mixtures} of Weibull distributions or exponential distributions
are used for failure time,  lifetime, wind speed,
forestry data, and so on.
For instance, \cite{Lawless2003}  applied a mixture of Weibull distributions
to the number of cycles before failure for a group of 60 electrical appliances.
\cite{Dwidayati2013}  used a mixture of Weibull distributions
for the lifetimes of breast cancer patients from 
medical records.
\cite{Zhang2001} applied a mixture of Weibull distributions to
the diameter distributions of rotated-sigmoid and uneven-aged stands.
See \cite{Castet2009}, \cite{Qin2012}, and \cite{Kao1959}
for more examples.

Suppose we have a set of independent and identically
distributed ({\iid}) observations,  $x_1, \ldots, x_n$,
from a two-component mixture
\begin{equation}
\label{model}
f(x; G)
=
\alpha_1 f(x; \btheta_1) + \alpha_2 f (x; \btheta_2).
\end{equation}
An elementary yet fundamental problem is the test of homogeneity:
\begin{equation*}
H_{0}:  \alpha_1 \alpha_2 (\btheta_1 - \btheta_2) = 0.
\label{testing.problem}
\end{equation*}

Research into homogeneity testing has a long history.
The earliest examples
include \cite{Hartigan1985} and \cite{Ghosh1985},
who revealed the nonstandard
asymptotic behavior of the likelihood ratio test.
A famous nonstandard approach is the C($\alpha$) test of
\cite{Neyman1966}. \cite{Bickel1993}, \cite{Chernoff1995}, 
\citet{Dacunha1999},  \cite{Chen2001}, and  \cite{LiuShao2003}
all contributed to the understanding of the classical likelihood ratio test
in the context of the finite mixture model.
Two waves of further development
led to the effective data analysis procedures
summarized in {the {\tt R} package {\tt MixtureInf}}.
One is the modified likelihood ratio test of \cite{Chen1998},
\citeauthor{Chen2001A} (\citeyear{Chen2001A}, \citeyear{Chen2004}),
and \citeauthor{Charnigo2004} (\citeyear{Charnigo2004}, \citeyear{Charnigo2010}).
Another is the EM-test; see \cite{Li2009}, \cite{Chen2009},
and \cite{Niu2011}.
{
Because of the additional nonregularities of location-scale mixtures 
in the form of the unbounded likelihood, 
the existing results are not applicable to general location-scale mixtures.}

We take up this task in this paper.
We show that the EM-test approach
remains effective for location-scale mixtures.
In Section \ref{main}, we develop an EM-test for homogeneity
tailored for location-scale mixtures, derive its limiting distribution,
and obtain its specific form in three location-scale mixtures.
In Section \ref{tuning}, we use an experimental approach to determine
a set of optimal tuning parameter values.
In Section \ref{simulation}, we show via simulation that the
proposed EM-test has close to nominal type I errors and
good power properties. In Section \ref{DataEx}, we give
two real-data examples. The paper ends with an Appendix
containing the technical derivations.

\setcounter{equation}{0} 
\section{Main results}\label{main}

The location-scale mixture is nonregular in several ways.
Given a set of {\iid} observations $x_1, \ldots, x_n$, the log-likelihood
function is given by
\begin{equation*}
\ell_n(G) =
\sum_{i=1}^{n}\log f(x_i; G).
\label{log.like}
\end{equation*}
When $G$ has only two support points, we also write
it as $\ell_n(\alpha_1, \alpha_2, \btheta_1, \btheta_2)$.
Let $\btheta_1 = (x_1, \sigma_1)^\tau$, $\btheta_2 = (0, 1)^\tau$,
and $\alpha_1 =\alpha_2 =0.5$ in $G$.
We have $f(x_1; G) \to \infty$ as $\sigma_1\to 0$ while
$f(x_i; G)$ has a finite lower bound for all {$i\neq1$}.
Hence, the log-likelihood $\ell_n(G)$ is unbounded.
This well-known undesirable property leads to the inconsistent
maximum likelihood estimation (MLE) of $G$ for location-scale mixtures.
To save the likelihood-based inference,
one may counter this aberration with a penalty function
in $\sigma_1$ and $\sigma_2$ similarly to \cite{Chen2008}
or a constraint as in \cite{Tanaka2009}.
As an alternative, we use
the penalized log-likelihood function
\bea
\tilde{\ell}_n(G)
&=&
\ell_n(G) +p(\alpha_1)+p(\alpha_2)+ p_n(\sigma_1) + p_n(\sigma_2)
\nonumber \\
&=&
\ell_n(G) +p(\alpha_1, \alpha_2)+ p_n(\sigma_1, \sigma_2),
\label{pen.log.like}
\eea
for some choice of $p(\cdot)$ and $p_n(\cdot)$.
Here, we have equated $p(\alpha_1)+p(\alpha_2)$ and $p(\alpha_1, \alpha_2)$
and so on for notational convenience.
We develop an effective EM-test based on
\eqref{pen.log.like} {in the next subsection}.

\vspace{1em}

\noindent{\bf 2.1~ EM-test}

{We first choose a set $\{\pi_1, \dots, \pi_J\} \in (0, 0.5]$,
for example $\{0.1, 0.3, 0.5\}$, as the initial values for $\alpha_1$ and
a positive integer $K$, for example $K=3$.}
{We then }define an EM-test statistic through the
following iteration steps:

\vs
\noindent
Step 1. Let $k=0$.
For a given $j$, let $\alpha_1^{(0)}=\pi_j$ and $\alpha_2^{(0)} = 1- \pi_j$.
Compute
$$
(\btheta_1^{(0)}, \btheta_2^{(0)})
=
\arg \max_{{\footnotesize\btheta_1, \btheta_2}}
\tilde{\ell}_n (\alpha_1^{(0)},\alpha_2^{(0)}, \btheta_1, \btheta_2).
$$

\vs
\noindent
Step 2. For $i=1, \dots, n$ and the current $k$, use an
{\bf E-step} to compute
\[
w^{(k)}_{i}
=
\dfrac{\alpha_1^{(k)}f(x_{i}; \btheta_1^{(k)})}
{\alpha_1^{(k)}f(x_{i}; \btheta_1^{(k)})
+
\alpha_2^{(k)}f(x_{i}; \btheta_2^{(k)})}.
\]
Update the parameters by an {\bf M-step} such that
\begin{eqnarray*}
(\alpha_1, \alpha_2)^{(k+1)}
 =
 \arg\max_{\alpha_1, \alpha_2}
 \left \{
 \sum_{i=1}^{n} w^{(k)}_{i} \log \alpha_1
 +
 \Big( n- \sum_{i=1}^{n}w^{(k)}_{i}\Big) \log \alpha_2
 +
 p(\alpha_1, \alpha_2)\right\}
\end{eqnarray*}
and
\begin{align*}
\btheta_1^{(k+1)}
&=
\arg\max_{\footnotesize\btheta}
\left
\{\sum\limits_{i=1}^{n}
w^{(k)}_{i}\log f(x_{i}; \btheta)+p_n(\sigma_1)
\right \},
\\
\btheta_2^{(k+1)}
&=
\arg\max_{\footnotesize\btheta}
\left \{\sum\limits_{i=1}^{n}
(1-w^{(k)}_{i}) \log f(x_{i}; \btheta)+p_n(\sigma_2) \right \}.
\end{align*}

Repeat the E-step and M-step $K-1$ times. Return
{$\left(\alpha_1^{(K)}, \alpha_2^{(K)}, \btheta^{(K)}_{1}, \btheta^{(K)}_{2}\right)$}.

\vs
\noindent
Step 3. Define
\[
M^{(K)}_n(\pi_j)
=
2\left\{\tilde{\ell}_n(\alpha_1^{(K)}, \alpha_2^{(K)}, \btheta^{(K)}_{1}, \btheta^{(K)}_{2})
-
\tilde{\ell}_n(0.5, 0.5,\hat \btheta_0, \hat \btheta_0)
\right\},
\]
where
$
\hat \btheta_0
=
\arg\max_{\footnotesize\btheta} \tilde{\ell}_n(0.5, 0.5, \btheta, \btheta)
$.

\vs
\noindent
Step 4.
Repeat Steps 1 to 3 for each $j=1, \ldots, J$.
Define the EM-test statistic to be
\be
\label{EMstat}
\EM^{(K)}_{n}=\max\left\{M_n^{(K)}(\pi_j): j = 1, \dots, J\right\}.
\ee

The null hypothesis $H_0$ is rejected if $\EM^{(K)}_{n}$
exceeds some critical value determined by its limiting distribution,
derived below.

\vspace{1em}
\noindent{\bf 2.2~ Asymptotic properties}

The EM-test statistic is location-scale invariant
when $p_n(\cdot)$ is invariant, and this can be achieved
{by} an appropriate choice. Therefore, without loss of generality,
we assume that under $H_0$, $\mu = 0$
and $\sigma= 1$. In other words, we take $f_0(x)$ as the
true distribution of $x_1, \ldots, x_n$ under the null hypothesis.

Two key quantities are involved in the asymptotic study:
the gradient vector and the Hessian matrix of
$f(x; \btheta)$ evaluated at $\btheta_0 = (0, 1)^\tau$.
They make up a vector of length 5, two partial derivatives
and three second-order partial derivatives (divided by 2)
with respect to $\btheta = (\mu, \sigma)$,
{
$$
\bb_{1i}
=  \left(\frac{\partial f(x_i; \btheta_0)/\partial \mu}{f(x_i; \btheta_0)},\frac{\partial f(x_i; \btheta_0)/\partial \sigma}{f(x_i; \btheta_0)}\right)^\tau
$$
and
$$
\bb_{2i}
=
\left(
\frac{\partial^2 f(x_i;\btheta_0)/\partial\mu^2}
{2f(x_i;\btheta_0)},
\frac{\partial^2f(x_i;\btheta_0)/\partial\mu\partial\sigma}{2f(x_i;\btheta_0)},
\frac{\partial^2 f(x_i;\btheta_0)/\partial\sigma^2}{2f(x_i;\btheta_0)}
\right)^\tau.
$$
}
Let $\bb_i = (\bb_{1i}^\tau, \bb_{2i}^\tau)^\tau$.
When $f_0(x)$ is sufficiently well-behaved,
$\bbE( \bb_i) = 0$, and well-defined $\bB = \var(\bb_i)$.
Let $\bB_{11}$, $\bB_{12}$, and $\bB_{22}$ be submatrices
of $\bB$ matching the partition $\bb_{1i}$ and $\bb_{2i}$
and let
${\tilde \bb}_{2i}=\bb_{2i}-{\bB_{21} \bB_{11}^{-1}}\bb_{1i}$.
We have
\(
\var({\tilde \bb}_{2i})
=
{\tilde \bB}_{22}
=\bB_{22}-\bB_{21}\bB^{-1}_{11}\bB_{12}
\)
and $\cov(\bb_{1i}, \tilde{\bb}_{2i}) = 0$.

Here is the main result, with the convention that
when $\bv = (v_1, v_2)^\tau$,
$$
\label{convention}
{\bv^2} = (v_1^2, 2v_1v_2, v_2^2)^\tau.
$$

\begin{theorem}
\label{thm2}
Suppose we have a random sample from model \eqref{model}
and the EM-test statistic is defined by \eqref{EMstat}
with the penalized likelihood function \eqref{pen.log.like}.
Assume Conditions B1--B3 on $f_0(x)$ and C1--C4  on
$p(\cdot)$, $p_n(\cdot)$; {these conditions are given in the Appendix}.
Let $\pi_1=0.5$ and $\pi_2, \ldots, \pi_J \in (0, 0.5)$.
Under the null hypothesis, for any fixed finite $K$,
as $n \rightarrow \infty$:\\
(i)
If 
$\bB_{11}$ has full rank
with $\tilde{\bB}_{22} (\bv^2)^\tau \neq 0$ for any $\bv \neq 0$
then
\be
\label{limit.distr1}
\EM^{(K)}_n
\cd
\sup_{\bv}
\left\{
2 (\bv^2)^\tau \bw - (\bv^2)^\tau \tilde \bB_{22} (\bv^2)
\right\},
\ee
where ${\bw}=(w_1, w_2, w_3)^\tau$ is a multivariate normal random vector with
mean zero and variance-covariance matrix ${\tilde \bB_{22}}$.

\noindent
(ii)
If $\bB_{11}$ has full rank, 
and the only null eigenvector of $ \tilde \bB_{22}$ has the form
$(u_1, 0, u_3)^\tau$ such that $u_1 u_3 > 0$, then
$$
\label{limit.distr2}
\EM_n^{(k)} \cd \chi_2^2.
$$
\end{theorem}

Conditions B1--B3 require $f_0(x)$ to be sufficiently smooth and to
satisfy some integration conditions. They are satisfied by most
commonly used distributions.
Conditions C1--C4 are requirements on the penalty functions.
Since the user can choose these functions,
the usefulness of the proposed EM-test is not affected.
Specific recommendations will be given later.

The conditions on the rank of $\bB$ correspond to the
strong identifiability introduced in \cite{Chen1995}.
Collinearity of the first, second-order, or even higher-order derivatives of the
component density function often leads to complex 
large-sample properties (\citealp{Ho2016}).

Because the value of $\bv^2 = (v_1^2, 2v_1v_2, v_2^2)^\tau$
is restricted to a two-dimensional manifold of $\cR^3$, the
limiting distribution in Case (i) is in general not the well-known
chi-square mixture. Nevertheless, its quantiles
are easily computed by Monte Carlo simulation.
When $\bB$ has a specific structure, as in case (ii),
the limiting distribution is particularly elegant. This is also the
case for the normal mixture model although it does not satisfy the
conditions specified in Case (ii);
{see \cite{Chen2009}. }

\vspace{1em}

\noindent{\bf 2.3~ Examples}

To illustrate the
proposed EM-test, we identify a few location-scale families satisfying
the conditions and work out their $\tilde {\bB}_{22}$ matrices.

\noindent
{\bf Logistic distribution}.
The density function of the standard logistic distribution is given by
\[
f_0(x)
=\frac{\exp(-x)}{ \{1+\exp(-x) \}^2}.
\]
To a sufficient numerical precision,
$$
\tilde \bB_{22}=
\left (
\begin{array}{ccc}
0.0063& 0  & -0.1043\\
0&   0.2062& 0\\
-0.1043 &  0  & 1.8498\\
\end{array}
\right ),
$$
which has full rank.
Hence, the logistic distribution belongs to Case (i),
and its null limiting distribution is given by \eqref{limit.distr1}
with the above $\tilde \bB_{22}$.

\noindent
{\bf Extreme-value distribution}.
The density function of the standard (type I) extreme-value distribution
is
\[
f_0(x)
=
\exp \{ x -\exp(x)\}.
\]
We find, to a sufficient numerical precision, that
$$
\tilde \bB_{22}=
\left(
\begin{array}{ccc}
0.3921& 0.9697& 1.1256\\
0.9697& 2.4928& 3.4362\\
1.1256& 3.4362& 7.8242\\
\end{array}
\right),
$$
which has full rank.
The extreme-value distribution also belongs to Case (i),
and its null limiting distribution is given by \eqref{limit.distr1}
with the above $\tilde \bB_{22}$.


\noindent
{\bf Student-t distribution}.
The density function of the standard student-t distribution with $v$
degrees of freedom is
$$
f_0(x)
=
\frac{\Gamma ( (\nu+1)/2)}{\Gamma (\nu/2 )\sqrt{\pi \nu}}
 \{1+  x^2/\nu\}^{- (\nu+1)/2 }.
$$
We consider the situation where $\nu$ is known.
We find its $\tilde{\bB}_{22}$ has rank 2 and {its null eigenvector $\bu=(u_1,u_2,u_3)^\tau$}
has $u_2 = 0$ while $u_1u_3 > 0$.
Thus, the EM-test statistic has a $\chi^2_2$ limiting
distribution under the null hypothesis of homogeneity.


\setcounter{equation}{0} 
\section{Experiments for tuning parameters}
\label{tuning}

To implement the EM-test, the user must select
penalty functions and tuning parameter values.
These choices affect the computational simplicity and
precision of the asymptotic distribution as well as the power properties
of the EM-test.
Similarly to \cite{Chen2009}, we suggest setting
$\{\pi_1, \ldots, \pi_J\} = \{0.1,0.3,0.5\}$ with $J=3$ and
$K=3$. We recommend
\be \label{pn.sigma}
\begin{multlined}
p(\alpha)
=
\log(1-|1 - 2\alpha|)
\\
p_n(\sigma)
=
-a_n\{\hat\sigma^2/\sigma^2+\log(\sigma^2/\hat\sigma^2)\}
\end{multlined}
\ee
for some $a_n>0$
with $\hat \sigma$ being the maximum likelihood estimator of
$\sigma$ under $H_0$.
This choice is equivalent to placing
a Gamma prior distribution on $\sigma^{-2}$.
The inclusion of $\hat \sigma^2$ makes the EM-test
location-scale invariant.
The specific functional forms enable
easy numerical computation.

The choice of $a_n$ influences the
type-I errors of the proposed test.
We take advantage of this property and
use experiments to recommend a value of $a_n$
to achieve accurate test sizes.


The experiment started with pilot trials
on many values of $a_n$ and the sample size $n$.
We found that when $a_n \leq 0.2$, the EM-test has
markedly inflated type-I errors compared to the nominal levels.
We then decided to run a $4 \times 4$ factorial design
for $a_n \in \{ 0.3, 0.4, 0.5, 0.6\}$ and
$n \in \{50, 100, 300, 500\}$
{and} to apply the data from the four location-scale mixtures as follows.

For each location-scale family,
we used the Monte Carlo method to obtain precise upper
quantiles for the limiting distributions of
$\EM^{(K)}_{n}$ in \eqref{EMstat}.
We used them to determine the rejection regions.
For each combination of $a_n$ and $n$
in the factorial design,
we generated 10000 random samples of size $n$
from $f_0(x)$ to obtain $\EM_n^{(3)}$ values
and therefore the rejection rate $\hat q$ at level $q$.
The discrepancy between $\hat q$ and $q$ is calculated as
\be
\label{discrep}
y=\log\{\hat q/(1-\hat q)\} - \log\{q/(1-q)\}.
\ee
The values for $q=0.05$ are given in
Table~\ref{err.experiment}
for the four location-scale families investigated.
We included only $t_{10}$ {for the student  $t$ distribution} for reality considerations.

\begin{table}[!htt]
\renewcommand{\arraystretch}{0.45}
\tabcolsep=1.5mm
\caption{\label{err.experiment}
Discrepancy between $\hat q$ and $q$ in terms of \eqref{discrep}
for four location-scale distribution families at $q=0.05$}
\begin{center}
\begin{tabular}{ccrrrr}
\hline
$a_n$&$n$&Logistic&Extreme&Student-t&Normal\\
\hline
0.3&  50&-0.1200&0.0270&-0.0234&-0.1778\\
0.4&  50&-0.2761&-0.1129&-0.2207&-0.4395\\
0.5&  50&-0.4115&-0.2897&-0.3664&-0.5845\\
0.6&  50&-0.5845&-0.3993&-0.5525&-0.7525\\
0.3&100&0.0413&0.1253&0.0146&-0.0106\\
0.4&100&-0.0561&0.0188&-0.1083&-0.1557\\
0.5&100&-0.1485&-0.0990&-0.2104&-0.2815\\
0.6&100&-0.2520&-0.1952&-0.3175&-0.3783\\
0.3&300&0.1328&0.1197&0.1804&0.0291\\
0.4&300&0.0753&0.0733&0.1366&-0.0256\\
0.5&300&0.0063&0.0188&0.0909&-0.0853\\
0.6&300&-0.0539&-0.0299&0.0393&-0.1509\\
0.3&500&0.0929&0.1328&0.0454&0.0126\\
0.4&500&0.0534&0.0851&0.0146&-0.0213\\
0.5&500&0.0209&0.0413&-0.0170&-0.0650\\
0.6&500&-0.0213&0.0000&-0.0517&-0.1037\\
\hline
\end{tabular}
\end{center}
\end{table}

The information from Table \ref{err.experiment} is utilized
in the following way. We first build a model for $y$
and a function of $n$ and $a_n$. Based on this model,
for each sample size $n$, we find a value
of $a_n$ such that the discrepancy $y$ between the observed
type-I error and the nominal level disappears.

After some exploratory analysis, we found that a linear regression of $y$
on $1/n$ and $\log(a_n-0.2)$ was satisfactory.
The covariate $\log(a_n-0.2)$ effectively confines the value of
$a_n$ in $(0.2,\infty)$, as suggested by our pilot study.
We next  regress $y$ in $1/n$ and  $\log(a_n-0.2)$.
Solving $\hat y=0$ leads to empirical formulas for $a_n$:
\begin{equation}
\label{an.formula}
a_n
 =
 \left\{
 \begin{array}{ll}
0.2+\exp(-0.959 - 119.899/n)& \mbox{Logistic}\\
0.2+\exp(-0.986 - 77.677/n)& \mbox{Extreme}\\
0.2+\exp(-1.032 - 103.737 /n) & \mbox{Student}$-$t \\
0.2+\exp(-1.410 - 114.433/n)  & \mbox{Normal} \\
    \end{array}
    \right..
 \end{equation}

We have implemented the EM-test using {\tt R} with these empirical formulas
for $a_n$
and the other suggested tuning parameters.
In the next section, we examine the
performance of the EM-test with the recommended parameters.

\setcounter{equation}{0} 

\section{Simulation}
\label{simulation}
The purpose of the simulation study is twofold.
First, we check if the limiting distribution of the
EM-test adequately approximates the finite-sample distribution.
Second, we compare the power of the EM-test with that of the likelihood ratio test (LRT).
Here, the LRT statistic is defined as
\begin{equation*}
M_n
=
2\{\ell_n(\tilde G) - \ell_n(\tilde G_0) \},
\label{Mn.stat}
\end{equation*}
{where
\(
\tilde G = \arg\max_{G \in \GG_2} \{  \ell_n(G)+p_n(\sigma_1, \sigma_2)\) \}
is the penalized maximum likelihood estimator of $G$ under the full model
and $\tilde G_0$ is the maximum likelihood estimator of $G$ under the null hypothesis.
The $p_n(\cdot)$ here is from (\ref{pn.sigma}) with $a_n=1/n$ to
prevent an unbounded log-likelihood;
$\GG_2$ is the parameter space for $G$ under the full model}.
The distributions for the LRT are simulated.

We generated data from various homogeneous distributions with a range of
sample sizes.
The rejection regions of the EM-test statistic $\EM_{n}^{(3)}$ are
based on the limiting distributions given in  Theorem \ref{thm2}.
The rejection rates for $10^5$ repetitions are
given in Table~\ref{typei} at three nominal levels.
Clearly, the type I error rates of the EM-test
are quite close to the nominal levels for all models and sample sizes.
Hence, the limiting distributions provide accurate approximations
for the finite-sample distributions of $\EM_{n}^{(3)}$
coupled with the recommended tuning parameters.

\begin{table}[!ht]
\renewcommand{\arraystretch}{0.45}
\tabcolsep=1.0mm
\caption{\label{typei}Simulated type I error rates for EM-test}
\begin{center}
\begin{tabular}{cc |cccccccccccc}
\hline
$f_0 $&Level&\multicolumn{12}{c}{$n$}\\
\hline
&&50&   75&  100&  200 & 300 & 400&  500 & 600 & 800& 1000& 3000& 5000\\
\hline
            &10\%&10.1&10.1&10.0&9.8&9.9&9.7&9.7&9.9&10.0&9.9&10.0&10.0\\
Logistic&  5\%&5.1&5.1&5.0&5.0&5.0&4.9&5.0&5.0&5.1&4.9&5.1&5.0\\
            &  1\%&1.1&1.0&1.0&1.0&1.0&1.0&1.0&1.0&1.1&1.0&1.0&1.0\\
            \hline
            &10\%&10.5&10.2&10.1&10.0&10.0&10.0&10.0&10.0&10.1&10.2&10.1&10.1\\
Extreme&  5\%&5.3&5.2&5.1&5.1&5.0&5.0&5.1&5.1&5.1&5.1&5.0&5.1\\
            &  1\%&1.1&1.1&1.1&1.0&1.0&1.0&1.1&1.0&1.0&1.0&1.1&1.0\\
            \hline
            &10\%&10.6&10.0&9.8&9.6&9.6&9.8&9.8&9.6&9.9&9.7&9.8&9.9\\
$t_{6}$&  5\%&5.3&5.2&5.0&4.8&4.9&4.9&5.0&4.8&4.9&4.9&4.9&5.0\\
            &  1\%&1.1&1.1&1.1&1.0&1.0&1.0&1.0&1.0&1.0&1.0&1.0&1.0\\
            \hline
            &10\%&10.2&9.8&9.7&9.6&9.7&10.1&9.9&10.0&9.9&9.8&10.0&10.1\\
$t_{10}$&  5\%&5.2&4.9&4.9&4.9&4.9&5.2&4.9&5.0&5.1&4.9&5.1&5.1\\
            &  1\%&1.1&1.0&1.0&1.0&1.0&1.1&1.0&1.0&1.1&1.0&1.1&1.1\\
            \hline
            &10\%&9.9&9.8&9.7&9.8&9.8&9.9&9.8&9.9&10.1&10.1&10.1&9.9\\
$t_{14}$&  5\%&5.0&4.9&4.8&4.9&5.0&5.0&5.1&5.0&5.2&5.1&5.2&5.0\\
            &  1\%&1.0&1.0&1.0&1.0&1.0&1.0&1.0&1.0&1.1&1.1&1.1&1.0\\
            \hline
            &10\%&10.1&10.3&10.1&9.9&9.9&10.2&10.2&10.2&10.0&10.2&10.1&10.1\\
$N(0,1)$&  5\%&5.1&5.2&5.2&5.0&5.0&5.1&5.1&5.2&5.1&5.2&5.1&5.2\\
            &  1\%&1.0&1.1&1.1&1.0&1.1&1.1&1.0&1.1&1.0&1.1&1.1&1.1\\
\hline
\end{tabular}
\end{center}
\end{table}

Next, we compare the power of the EM-test with the LRT
under the logistic, Weibull, and $t_{6}$ kernels
for two sample sizes: $n=200$ and $n=400$.
The models and the simulated powers of the EM-test
and the LRT at the 5\% nominal level
are presented in Tables \ref{power1}--\ref{power3}.
{The simulated powers are calculated from $10^4$ repetitions.}
For a fair comparison, the rejection regions 
are based on $10^5$ random samples  from the null model.
It can be seen that the EM-test is much more powerful than the LRT in almost
all cases.
When the mixing proportions are $0.05$ and $0.95$, the LRT
is occasionally slightly more powerful.

\begin{table}[!htt]
\renewcommand{\arraystretch}{0.45}
\tabcolsep=1.5mm
\caption{\label{power1}Simulated powers for EM-test and LRT
for logistic mixtures at 5\% nominal level}
\begin{center}
\begin{tabular}{l| cccc }
\hline
Alternative model&$\EM_{n}^{(3)}$&LRT&$\EM_{n}^{(3)}$&LRT\\
\hline
&\multicolumn{2}{c}{$n=200$}&\multicolumn{2}{c}{$n=400$}\\
L1: $0.5\{(0,1)\}+0.5\{(3.0,1.0)\}$    &63.0&34.1&92.0&68.3\\
L2: $0.5\{(0,1)\}+0.5\{(2.0,2.0)\}$    &71.0&50.5&95.5&83.3\\
L3: $0.5\{(0,1)\}+0.5\{(0, 2.3)\}$       &57.7&40.6&88.3&70.4\\
L4: $0.8\{(0,1)\}+0.2\{(3.0,1.0)\}$    &46.3&25.6&78.7&52.7\\
L5: $0.8\{(0,1)\}+0.2\{(2.0, 2.0)\}$    &69.7&54.1&95.3&86.0\\
L6: $0.8\{(0,1)\}+0.2\{(0, 2.3)\}$       &58.1&46.6&88.2&74.7\\
L7: $0.95\{(0,1)\}+0.05\{(5.0,1.0)\}$&45.3&37.2&78.5&68.8\\
L8: $0.95\{(0,1)\}+0.05\{(3.5,2.0)\}$&33.1&33.1&60.7&54.6\\
L9: $0.95\{(0,1)\}+0.05\{(0, 3.5)\}$   &52.0&57.4&79.7&79.9\\
\hline
\end{tabular}
\end{center}
\end{table}

\begin{table}[!htt]
\renewcommand{\arraystretch}{0.45}
\tabcolsep=1.5mm
\caption{\label{power2}
Simulated powers for EM-test and LRT for
extreme-value mixtures at 5\% nominal level}
\begin{center}
\begin{tabular}{l| cccc }
\hline
Alternative model&$\EM_{n}^{(3)}$&LRT&$\EM_{n}^{(3)}$&LRT\\
\hline
&\multicolumn{2}{c}{$n=200$}&\multicolumn{2}{c}{$n=400$}\\
\hline
E1: $0.5\{(0,1)\}+0.5\{(1.8,1.0)\}$    &69.5&42.9&94.5&77.9\\
E2: $0.5\{(0,1)\}+0.5\{(1.3,1.2)\}$    &64.3&39.2&92.5&73.1\\
E3: $0.5\{(0,1)\}+0.5\{(0, 2.0)\}$       &57.8&39.6&87.6&70.7\\
E4: $0.8\{(0,1)\}+0.2\{(1.4,1.0)\}$    &70.8&47.5&95.2&82.1\\
E5: $0.8\{(0,1)\}+0.2\{(1.0,1.2)\}$    &60.1&40.8&89.7&73.1\\
E6: $0.8\{(0,1)\}+0.2\{(0, 2.0)\}$       &67.3&54.0&92.9&83.0\\
E7: $0.95\{(0,1)\}+0.05\{(1.4,1.0)\}$&37.9&28.5&66.2&52.7\\
E8: $0.95\{(0,1)\}+0.05\{(1.0,1.2)\}$&25.9&20.4&45.2&35.8\\
E9: $0.95\{(0,1)\}+0.05\{(0, 2.0)\}$   &26.8&25.4&44.4&40.7\\
\hline
\end{tabular}
\end{center}
\end{table}

\begin{table}[!ht]
\renewcommand{\arraystretch}{0.45}
\tabcolsep=1.5mm
\caption{\label{power3}
Simulated power for EM-test and LRT for $t_6$ mixtures
at 5\% nominal level}
\begin{center}
\begin{tabular}{l | cccc }
\hline
Alternative model&$\EM_{n}^{(3)}$&LRT&$\EM_{n}^{(3)}$&LRT\\
\hline
&\multicolumn{2}{c}{$n=200$}&\multicolumn{2}{c}{$n=400$}\\
\hline
T1: $0.5\{(0,1)\}+0.5\{(1.8,1.0)\}$      &49.1&21.9&81.0&45.7\\
T2: $0.5\{(0,1)\}+0.5\{(2.0,1.5)\}$         &71.9&42.5&95.6&77.3\\
T3: $0.5\{(0,1)\}+0.5\{(0,2.5)\}$        &63.3&42.3&92.8&75.9\\
T4: $0.8\{(0,1)\}+0.2\{(2.5,1.0)\}$     &88.3&69.2&99.5&96.4\\
T5: $0.8\{(0,1)\}+0.2\{(2.0,1.5)\}$     &71.2&45.9&95.8&81.3\\
T6: $0.8\{(0,1)\}+0.2\{(0,2.5)\}$         &59.6&44.7&90.1&73.9\\
T7: $0.95\{(0,1)\}+0.05\{(3.0,1.0)\}$ &30.7&21.1&58.7&42.3\\
T8: $0.95\{(0,1)\}+0.05\{(3.0,2.0)\}$ &40.1&35.6&69.7&62.2\\
T9: $0.95\{(0,1)\}+0.05\{(0,3.5)\}$   &27.9&35.0&51.9&54.0\\
\hline
\end{tabular}
\end{center}
\end{table}

\setcounter{equation}{0} 

\section{Data examples}
\label{DataEx}
We now examine the performance of the EM-test via two
real-data examples.
The first data set concerns the maximum precipitation in
24 hours in Montreal from 1872--2017.
The daily precipitations in Montreal are available from
weatherstats.ca based on Environment and Climate Change Canada data.
We calculate the maximum precipitation in 24 hours (in mm) for each year.
The figures are incomplete in 1873 and 1993,
and hence the observations for those two years are missing.
In total, we have 144 observations.

\cite{Shoukri1988} proposed using the log-logistic distribution
to model the maximum precipitation in 24 hours.
For illustration, we apply the EM-test to the maximum precipitation data
to check for potential heterogeneity through a test of homogeneity.
We log-transform the 144 observations before the EM-test is applied. 
{With the logistic distribution  being the component distribution},
the value of the EM-test statistic is found to be 6.290 with a p-value
of 0.043, calibrated by its limiting distribution.
For comparison, we also calculate the LRT, which is found to be 10.574.
Since both the EM-test and the LRT are invariant to the location and scale transformation,
we obtain their finite-sample distributions by generating
$10^5$ random samples from the standard logistic distribution.
Calibrated by their respective finite-sample distributions,
the p-values of the EM-test and LRT
are found to be 0.043 and 0.072, respectively.
Note that the finite-sample distribution and the limiting distribution of the
EM-test give the same p-value to the third decimal place.
Based on the p-value, the EM-test speaks more forcefully about
the presence of heterogeneity.
Indeed, the LRT fails to reject the homogeneous model at the 5\% level,
but the EM-test detects heterogeneity at the 5\% level.

%
%

Some related statistics for this data set are as follows.
The penalized maximum likelihood estimator of
the mixing distribution is given by
\[
\hat G
=
0.134\{(3.803, 0.124)\} + 0.866\{(4.307, 0.071)\}.
\]
Figure~\ref{preci} gives a histogram of the 144 maximum precipitation values
 along with the homogeneous logistic fitting
and the mixture of two logistic distributions fitting.
Clearly, the mixture
successfully captures the mode around 4.3, but
the homogenous logistic fitting does not.

%

\begin{figure}[!ht]
\centerline{\includegraphics[scale=0.5,angle=-90]{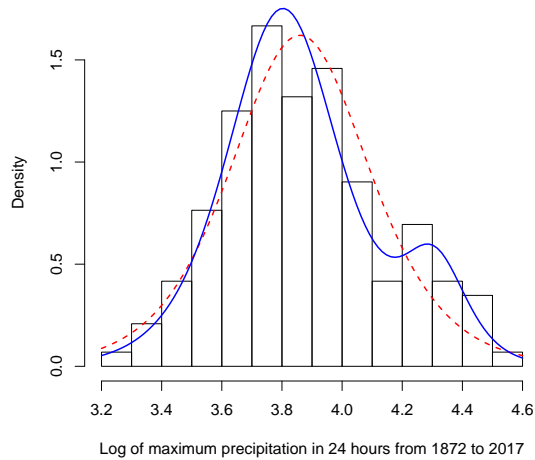}}
\caption{\label{preci}
Histogram of the maximum-precipitation data along with
the homogeneous logistic fitting (dashed red line) and
the mixture of two logistic distributions fitting (solid blue line).
}
\end{figure}

The second data set is from Example 3.4.1 of  \cite{Lawless2003}
who suggested a mixture of two Weibull distributions
for the number of cycles to failure for a group of 60 electrical appliances.
\cite{Lawless2003} argued that this mixture
provides an adequate fit to the 60 failure times
since its fit to the survival function
is quite close to the  Kaplan--Meier estimate of that function.
We apply the EM-test to the 60 log-transformed observations
for homogeneity with extreme-value kernel distributions.
The EM-test statistic is found to be 6.595 with p-value 0.037,
calibrated by its limiting distribution.
For comparison, we also calculate the value of the LRT,
which is found to be 9.669.
Since both the EM-test and the LRT are invariant to the location and
scale transformation,
we obtain their finite-sample distributions
by simulation with $10^5$ data sets.
Calibrated by their respective finite-sample distributions,
the p-values are found to be 0.038 and 0.080, respectively.
Again, the p-values from the finite-sample distribution and
the limiting distribution are quite close.
Further, the EM-test provides stronger evidence for
rejecting the homogeneous model.

%
%



\renewcommand\theequation{A.\arabic{equation}} 
\setcounter{equation}{0} 

\section*{Appendix}

Suppose $\bar G$ is a mixing distribution {with two support points}.
Let its support points be
$\bar{\btheta}_h = (\bar \mu_h, \bar{\sigma}_h)^\tau$, $h=1,2$,
and its mixing proportions $\bar{\alpha}_1$ and $\bar{\alpha}_2$.
This convention extends to $\hat{G}$, $\hat{G}^{(k)}$, and so on.
We use $G_0$ for $G$ when
$\btheta_1 = \btheta_2 = \btheta_0 = (0, 1)^\tau$.
Because the EM-test is location-scale invariant,
we assume that $G_0$ is the null mixing distribution.
Note that $f(x; G_0) = f_0(x)$.

Here are the conditions under which the various asymptotic
results are proved.

\vspace{1ex}
\noindent
{\bf B1}. (Smoothness and integrability)
$f_0(x)$ is five times continuously differentiable with respect to $x$.
For $k =0, 1, 2, 3, 4$,
\[
\bbE \big | \log f^{(k)} _0(X) \big | < \infty;
\]
and for all $\btheta = (\mu, \sigma)$ sufficiently close to $\btheta_0 = (0, 1)$,
\[
\left | \log f_0^{(5)} \left(\frac{x-\mu}{\sigma}\right) \right | \leq g(x)
\]
for some $g(x)$ such that $\bbE \{ g(X)\} < \infty$.

\noindent
{\bf B2}. (Identifiability)
For any two mixing distributions $G_1$ and $G_2$ with
at most two support points,
\(
f(x; G_1) = f(x; G_2)
\)
for all $x$ implies $G_1= G_2$.

\noindent
{\bf B3}. (Tail condition)
For any $0 < a \leq 1$, $f_0(x) \leq f_0(a x)$
and $f_0(x)$ satisfies $\sup_x (1+ x^2) f_0(x) < \infty$.

Conditions B1 and B2 are natural requirements for
ensuring manageable asymptotic properties of
the likelihood ratio statistics.
Condition B3 implies that the density function is unimodal
and the mode is at $x =0$.
If the mode of $f_0$ is at $x = x_0 \neq 0$, then we may
simply replace $f_0(x)$ by $\tilde f_0(x) = f_0(x - x_0)$
in the definition of the mixture model.
All our examples satisfy these conditions.

Next, we place some conditions on
the penalty functions $p(\alpha)$ and $p_n(\sigma)$.

\noindent
{\bf C1}. $p(\alpha)$ is continuous, maximized at $\alpha=0.5$,
and goes to negative infinity as $\alpha \rightarrow 0$.
Without loss of generality, $p(0.5) = 0$.

\noindent
{\bf C2}. $\sup_{\sigma>0} \{p_n(\sigma)\}^+=o(n)$,
$p_n(\sigma)=o(n)$, and $p'_n(\sigma)=o_p(n^{1/2})$ for all $\sigma$.

\noindent
{\bf C3}. $p_n(\sigma)\leq (\log n)^2\log(\sigma)$,
when $0 < \sigma \leq 1/n$ and $n$ is large.

\noindent
{\bf C4}. The penalty on $\sigma$ is scale-invariant:
namely, for any nonrandom constants $a > 0$ and $b$, the data-dependent penalty
{$p_n(a \sigma; ax_1 + b,\ldots, ax_n + b)  = p_n(\sigma; x_1,\ldots,x_n)$.}


These conditions serve as guidelines for
choosing the penalty functions. They are not restrictive
as long as such functions exist. Mathematically,
C1 makes $\alpha = 0.5$ {the} preferred value through $p(\cdot)$.
Conditions C2 and C3 prevent the penalties from
taking over the likelihood, and they discriminate against models
with small $\sigma$ values.  Condition C4 is not needed
for asymptotic considerations
but it ensures location-scale invariance.

\subsection*{Some  lemmas}

We first establish some properties of the point estimators.
Lemma \ref{lem1} gives a result on the order
of some $\bar G$ satisfying certain properties.
Let $\bm_1$ and $\bm_2$ be vectors of ``centered''
first and second moments of $G$,
and $\bm^\tau = (\bm^\tau_1, \bm^\tau_2)$:
\begin{align*}
\bm_1 = &
\alpha_1 (\btheta_1 - \btheta_0)  + \alpha_2 (\btheta_2 - \btheta_0),
\\
\bm_2 = &
\alpha_1 (\btheta_1 - \btheta_0)^2  + \alpha_2 (\btheta_2 - \btheta_0)^2.
\end{align*}
{Here we have used in the definition of $\bm_2$ the convention that
when $\bv = (v_1, v_2)^\tau$,
\(
{\bv^2} = (v_1^2, 2v_1v_2, v_2^2)^\tau.
\)}

\begin{lemma}
\label{lem1}
Assume the conditions of Theorem \ref{thm2}.
Let $\bar G$ be any estimator of $G$ such that
$\bar \alpha_1, \bar \alpha_2 \in [\delta, 1-\delta]$ for some $\delta \in (0,0.5)$
and for some constant $c$,
\[
\tilde{\ell}_n(\bar G)-\tilde{\ell}_n(G_0 ) > c >-\infty.
\]
Then,
for both Cases (i) and (ii)  specified in Theorem \ref{thm2},
\[
\bar{\btheta}_h - \btheta_0= O_p(n^{-1/4}), ~{h=1,2}; ~~
{\bar \bm}_1 = O_p(n^{-1/2}).
\]
\end{lemma}

\begin{proof}
Assume $\bar{\btheta}_h - \btheta_0 = o_p(1)$ under the lemma conditions.
Because the proof of this claim is tedious, we will present it separately.
With this assumption, we define
\bea
R_{1n}(\bar G)
&=&
2\{\ell_n(\bar G)- \ell_n(G_0)\} +
2 \{p_n(\bar\sigma_1, \bar\sigma_2) - p_n(1,1) + p(\bar\alpha_1, \bar \alpha_2) )\}
\nonumber \\
\label{penal.LR}
&=&
r_{1n}(\bar G) + O_p(1)
\eea
with
\[
r_{1n}(\bar{G}) = 2\{\ell_n(\bar G)- \ell_n(G_0)\}
= 2\sum_{i=1}^n\log(1+\delta_i)
\]
and
\[
\delta_i=
\big [
\{ \bar\alpha_1 f(x_{i};\bar{\btheta}_1)
+
\bar\alpha_2 f(x_{i}; \bar{\btheta}_2 ) \}
-
f_0(x_{i})
\big ]
{\big/}f_0(x_{i}).
\]

Expanding $f(x_i; \btheta)$ at $\btheta_0$, we obtain
\begin{equation}
\label{delta.linear}
\delta_i
=
\bar \bm_1^\tau \bb_{1i} + \bar \bm_2^\tau \bb_{2i} + \varepsilon_{in}
=
\bar \bm^\tau \bb_i + \varepsilon_{in}
\end{equation}
where the $\varepsilon_{in}$ denote remainders.
For $\btheta$ such that $\btheta - \btheta_0$ is very small,
\begin{equation}
\label{remainder}
\sum_{i=1}^n\varepsilon_{in}=o_p(1 + n\|{\bar \bm}\|^2)
\end{equation}
by referring to similar proofs given by \cite{Chen2001A} and \cite{Chen2001}.

Next, we use \eqref{delta.linear} and
\begin{equation}
r_{1n}(\bar{G})
\leq
2\sum_{i=1}^n \delta_i - \sum_{i=1}^n\delta_i^2
+ (2/3)\sum_{i=1}^n\delta_i^3
\label{inequality.R1n}
\end{equation}
to develop an upper bound for $R_{1n}(\bar G)$.
By some  straightforward algebra, we have
\begin{equation}
\label{delta.square.cubic}
\begin{cases}
\sum_{i=1}^n\delta_i^2
=
\sum_{i=1}^n (\bar{\bm}^\tau \bb_i )^2 +o_p(1 + n\|{\bar \bm}\|^2),
\\
\sum_{i=1}^n\delta_i^3
=
\sum_{i=1}^n (\bar{\bm}^\tau \bb_i )^3 +o_p(1 + n\|{\bar \bm}\|^2).
\end{cases}
\end{equation}

Note that
\(
{\bar \bm}^\tau {\bB} {\bar \bm}
=
\bar \bm_1^{*\tau} \bB_{11} \bar \bm_1^*
+ \bar \bm_2^\tau {\tilde \bB}_{22}\bar \bm_2
\)
where $\bar \bm_1^* = \bar \bm_1 - \bB_{11}^{-1} \bB_{12} \bar \bm_2$.
{Because either $\bB$ has full rank
or $\bar m_2$ is not in the null space of $\tilde \bB_{22}$,
we have ${\bar \bm}^\tau {\bB} {\bar \bm} > 0$ when $\bar \bm \neq 0$.}

The positive definiteness and the law of large numbers imply
\be
\begin{cases}
n^{-1} \sum_{i=1}^n (\bar{\bm}^\tau \bb_i )^2
=
{\bar \bm}^\tau {\bB} {\bar \bm} \{1+o_p(1)\},        \\
n^{-1} \sum_{i=1}^n (\bar{\bm}^\tau \bb_i )^3
= o_p(1) .
\end{cases}
\label{square.cubic}
\ee
Using Condition C1 and combining \eqref{penal.LR}--\eqref{square.cubic}, we have
\ba
R_{1n}(\bar G)
&\leq &
2{\bar \bm}^\tau\sum_{i=1}^n\bb_i
-
n{\bar \bm}^\tau {\bB}{\bar \bm}\{1+o_p(1)\}
+
o_p(1 + n\|{\bar \bm}\|^2)
\nonumber \\
&=&
2 {\bar \bm^{*\tau}_1}\sum_{i=1}^n\bb_{1i}
-
n{\bar \bm^{*\tau}_1}{\bB_{11}}{\bar \bm^*_1}
\nonumber\\
&&
+ 2{\bar \bm_2}^\tau \sum_{i=1}^n{\tilde \bb}_{2i}
-
n {\bar \bm_2}^\tau{\tilde \bB_{22}} {\bar \bm_2}
+
o_p(1 + n\|{\bar \bm}\|^2).
\label{upper.bound2}
\ea
Unless both ${\bar \bm^{*}_1}=O_p(n^{-1/2})$ and ${\bar \bm_2}=O_p(n^{-1/2})$,
this upper bound will go to $- \infty$, which contradicts the
lemma assumption.
When this is the case and
$\bar \alpha_1, \bar \alpha_2 \in [\delta, 1 - \delta]$,
we must have ${\bar \bm_1}=O_p(n^{-1/2})$,
and both $\bar{\btheta}_1 - \btheta_0 = O_p(n^{-1/4})$
and $\bar{\btheta}_2 - \btheta_0 = O_p(n^{-1/4})$.
\end{proof}

Let $\bar G$ be estimators of  $G$ as before and
\[
\bar \omega_i
=\dfrac{\bar\alpha_2 f(x_{i}; {\bar \btheta}_2)}
{ {\bar \alpha}_1 f(x_{i}; \bar\btheta_1)+{\bar \alpha}_2 f(x_{i}; {\bar\btheta}_2)}.
\]
Define
\[
H_n(\alpha_1)
=
\left(n-\sum\limits_{i=1}^{n}\bar \omega_i\right)\log \alpha_1+
\sum\limits_{i=1}^{n}\bar \omega_i \log(\alpha_2) +p(\alpha_1, \alpha_2).
\]
The EM-test updates the mixing proportions via
$\bar\alpha_1^{*} = \arg\max_{\alpha} H_n(\alpha_1)$.
The following lemma claims that when the null model is true,
$\bar\alpha_1^{*}$ stays close to $\alpha_1$ after a single
EM-iteration. The proof is identical to one in
\cite{Li2009}, so it is omitted.

\begin{lemma}
\label{lem3}
Under the conditions of Lemma \ref{lem1},
if $\bar \alpha_1 - \alpha_1=o_p(1)$ for some $\alpha_1 \in (0,0.5]$,
then $\bar\alpha_1^{*}-\alpha_1=o_p(1)$.
\end{lemma}

\begin{theorem}
\label{thm1}
Assume the Conditions of Theorem \ref{thm2} and the null distribution.
Let $G^{(k)}$ be the intermediate $G$ obtained with
the starting mixing proportion $\alpha$ for $\alpha_1$ after $k$ iterations.
Then
\[
\btheta_1^{(k)} - \btheta_0 = O_p (n^{-1/4}), ~~
\btheta_2^{(k)} - \btheta_0 = O_p (n^{-1/4}), ~~
\bm_1^{(k)} = O_p(n^{-1/2}).
\]
\end{theorem}

\noindent
\begin{proof}
The EM-algorithm has the property that the likelihood
increases after each iteration even with penalty terms
(\citealp{Dempster1977}; \citealp{Wu1983}). Hence, for any $k\leq K$,
\[
\tilde{\ell}_n(G^{(k)})
\geq
\tilde{\ell}_n(G^{(1)})
\geq
\tilde{\ell}_n(G_0).
\]
Therefore,
\[
\tilde{\ell}_n(G^{(k)}) - \tilde{\ell}_n(G_0) \geq c > -\infty.
\]
Hence, by Lemmas \ref{lem1} and \ref{lem3}, {$G^{(k)}$} has these properties.
\end{proof}

Here is some preparation for the proof of Theorem  \ref{thm2}. Let
${\bar \bv}
=\sqrt{\bar\alpha_1/{\bar\alpha}_2} (\bar{\btheta}_1 - \btheta_0)
$.
We have
\begin{align*}
\bar \bm_2 - \bar \bv^2
=&
\bar \alpha_2^{-1}
\{ \bar \alpha_2^2 (\bar \btheta_2 - \btheta_0)^2
-
 \bar \alpha_1^2 (\bar \btheta_1 - \btheta_0)^2 \}
 \\
=&
 \bar \alpha_2^{-1} \bar{\bm}_1
\{
\bar \alpha_2 (\bar \btheta_2 - \btheta_0)
-
\bar \alpha_1 (\bar \btheta_1 - \btheta_0)
\}
= \bar{\bm}_1 o_p(1).
\end{align*}
In addition,
\bea
2{\bar \bm_2}^\tau \sum_{i=1}^n{\tilde \bb}_{2i}
-
n{\bar \bm_2}^\tau{\tilde \bB_{22}}{\bar \bm_2}
&=&
2{(\bar \bv^2)^\tau} \sum_{i=1}^n{\tilde \bb}_{2i}
-
n{(\bar \bv^2)}^\tau{\tilde \bB_{22}}{\bar \bv^2}
\nonumber \\
&&
+
o_p(1 + n\|{\bar \bm}\|^2).
\label{side}
\eea

\noindent
{\bf Proof of Theorem \ref{thm2}.}

First, we consider the case where the EM-iteration starts
from $\pi_1 = 0.5$. We write its outcome as $\hat G_{0.5}$.
Let
$R_{0n}
=
2\{\tilde{\ell}_n(\hat G_{0.5})
-
\tilde{\ell}_n(G_0)\}$.
A classical result concerning regular models (\citealp{Serfling1980})
states that
\[
R_{0n}
=
n^{-1}
\big \{ \sum_{i=1}^n \bb_{1i} \big \}^\tau
\bB_{11}^{-1} \big \{ \sum_{i=1}^n \bb_{1i} \big \}
+o_p(1).
\]
Hence,
\begin{equation*}
2{\bar \bm^{*\tau}_1}
\sum_{i=1}^n \bb_{1i}
-n{\bar \bm^{*\tau}_1}{\bf B_{11}}{\bar \bm^*_1}\{1+o_p(1)\}
\leq
R_{0n}+o_p(1).
\label{LLR1}
\end{equation*}
Under the theorem conditions, and with \eqref{side}, we have
\ba
R_{1n}(G^{(k)})
&\leq &
R_{0n}
+\sup_{\bm_2}
\big \{
2{\bm_{2}^\tau}\sum_{i=1}^n{\tilde \bb}_{2i}
-
n{\bm_2^\tau}{\bf\tilde B_{2}}{\bm_2}
\big \} +o_p(1) \\
&=&
R_{0n}
+ \sup_{\bf v}
\big \{
2{(\bv^2)^\tau}\sum_{i=1}^n{\tilde \bb}_{2i}
-n{(\bv^2)^\tau}{\bf\tilde B_{22}}{\bv^2}
\big \} +o_p(1).
\ea
Further, by the definition of $M_n^{(k)}(\pi_j)$,
we have
\begin{equation*}
M_n^{(k)}(\pi_j) = R_{1n}(G^{(k)}) - R_{0n}
\leq
\sup_{\bf v}
\big \{
2{(\bv^2)^\tau}\sum_{i=1}^n{\tilde \bb}_{2i}
-n{(\bv^2)^\tau}{\bf\tilde B_{22}}{\bv^2}
\big \}
+ o_p(1).
\end{equation*}
The leading term on the right-hand side
does not depend on $\pi_j$, so
\begin{equation*}
\EM_n^{(K)}
\leq
\sup_{\bf v}
\big \{
2{(\bv^2)^\tau}\sum_{i=1}^n{\tilde \bb}_{2i}
-n{(\bv^2)^\tau}{\bf\tilde B_{22}}{\bv^2}
\big \} + o_p(1).
\label{case.I.upperbound}
\end{equation*}

Next, we show that the above inequality can be
tightened to equality.
Since the EM-iteration increases the penalized
likelihood (\citealp{Dempster1977}), we need only show
this result when $k=1$.
It suffices to find a $\hat G$ at which the upper bound is attained.
Let
\begin{align*}
&{\hat \bv}
=
\arg\sup_{\bv}
\left \{
2{(\bv^2)^\tau}\sum_{i=1}^n{\tilde \bb}_{2i}
-n{(\bv^2)^\tau}{\bf\tilde B_{22}}{\bv^2}
\right \},
\\
&{\hat \bm_1}
=
(n{\bB_{11}})^{-1}\sum_{i=1}^n\bb_{1i}
+
{\bB_{11}^{-1}\bB_{12}}{\hat \bv^2}.
\end{align*}
Further, let $\hat\alpha_1 = \hat \alpha_2 =0.5$,
$\hat\mu_1 = \hat v_1$, and $\hat \sigma_1 = \hat v_2 + 1$.
Regard $(\mu_2, \sigma_2)$ as variables in the
equation
\begin{equation*}
\hat \bm_1=
\begin{pmatrix}
\hat\alpha_1 \hat\mu_1+\hat \alpha_2 \hat\mu_2\\
\hat\alpha_1 (\hat\sigma_1-1)+\hat\alpha_2 (\hat\sigma_2-1)
\end{pmatrix},
\label{hatm1}
\end{equation*}
and let its solution be $(\hat \mu_2,\hat \sigma_2)$.
The solutions clearly satisfy
$\hat\mu_h=O_p(n^{-1/4})$,
$\hat\sigma_h=O_p(n^{-1/4})$, $h=1,2$.
Based on this order assessment,  we get
\begin{align*}
\EM_n^{(K)}
\geq
&~
M_n^{(1)}(0.5)
\geq
R_{1n}(\hat G_{0.5})-R_{0n}\\
=
&~
\sup_{\bv}
\big \{
2{(\bv^2)^\tau}\sum_{i=1}^n{\tilde \bb}_{2i}
-n{(\bv^2)^\tau}{\bf\tilde B_{22}}{\bv^2}
\big \}
+o_p(1).
\end{align*}
Since
$n^{-1/2}\sum_{i=1}^n{\tilde \bb}_{2i}
\to N(0, \tilde \bB_{22})$ in distribution, we get
\[
\EM_n^{(K)}
\to
\sup_{\bv}
\{
2 (\bv^2)^\tau \bw - (\bv^2)^\tau \tilde \bB_{22} (\bv^2)
\}
\]
for some multivariate normal random vector $\bw$
as given in the theorem.
This completes the proof of Case (i).

When $\bB_{22}$ has rank 2 as specified in Case (ii),
we must have $w_3 = a w_1$ for some $a < 0$.
Let $\bt = \left( v_1^2 + a v_2^2, 2v_1 v_2 \right)^\tau$,
{$\bw^* = (w_1, w_2)^\tau$},
and $\Sigma = \var( \bw^*)$.
The limit of $\EM_n^{(K)}$ is
\[
\sup_{\bt}
\{
2 \bt^\tau \bw^* - \bt^\tau \Sigma \bt
\}
\leq
\bw^* \Sigma^{-1} \bw^*,
\]
and equality holds
if $\Sigma^{-1} \bx = \left( v_1^2 + a v_2^2, 2v_1 v_2 \right)^\tau$
has a solution in $\bv$. The solution exists because
$\bt = \left( v_1^2 + a v_2^2, 2v_1 v_2 \right)^\tau$
can take any values in $\cR^2$.
Clearly, $\bw^* \Sigma^{-1} \bw^*$ has a $\chi_2^2$ distribution.
Hence, the limiting distribution in Case (ii) has the simpler
form given.

\subsection*{Proof of consistency}
A missing piece in the proof of Theorem \ref{thm2} is that
$\bar G$ satisfies
\(
\bar \btheta_h - \btheta_0 = o_p(1)
\)
for $h=1, 2$.
Since $\bar \alpha_1$ is bounded away from both
$0$ and $1$ by design, the above claim is implied should
$\bar G$ be consistent. Consistency of $\bar G$
in turn is implied by general consistency of
{the penalized
MLE}, a topic discussed by \cite{Chen2008},
\cite{Tanaka2009}, and \cite{Chen2017} in similar
situations.
Since consistency itself is not the focus of this paper,
we give a nonrigorous proof aided by
intuition. We plan to develop a full proof in the future.

\begin{lemma}
\label{lemA1}
{Let $x_1, \ldots, x_n$ be {\iid} random observations from  $f_0(x)$
with $\sup_x f_0(x) = M_0 = 1$.
Then
\[
\sup_{\mu}
| F_n(x + \epsilon) - F_n(x) |
\leq
2 \epsilon+10 n^{-1} \log n
\]
holds uniformly for all $\epsilon>0$ almost surely.
Here $F_n(x)$ is the empirical cumulative distribution function of $x_1, \ldots, x_n$.}
\end{lemma}

Remark: A scale transformation will make $M_0 = 1$,
which simplifies the presentation.
This result can be found in \cite{Chen2017}.

\begin{lemma}
\label{lemA2}
Assume Conditions B1--B3.
Let $G = \alpha_1 \{ \btheta_1\} + \alpha_2 \{\btheta_2\}$
be a mixing distribution with $\sigma_1 \leq \sigma_2$.
For some positive constants $\delta_0, \epsilon_0$, define
\[
g(x; G)
=
\delta_0 + \frac{\alpha_1}{\epsilon_0} f_0 \left ( \frac{x - \mu_1}{\epsilon_0} \right)
+
\frac{\alpha_2}{\sigma_2} f_0 \left ( \frac{x - \mu_2}{\sigma_2} \right).
\]
{Then, for a sufficiently small $\epsilon_0$, when $\sigma_1 < \epsilon_0$,
we have\\
(i)  for any $x$,
\be
\label{lemma4.1}
\log f(x; G) \leq - \log \sigma_1 + \log \{g(x; G) \};
\ee

\noindent (ii)
for any $x$ satisfying
$|x - \mu_1| \geq {\sigma_1^{2/3}}$,
\be
\label{lemma4.2}
\log f(x; G) \leq \log g(x; G).
\ee
}
\end{lemma}

\begin{proof}
{When  $\sigma_1 < \epsilon_0$,
$f_0 \big( (x - \mu_1)/\epsilon_0\big)
\geq
f_0 \big( (x - \mu_1)/\sigma_1\big)
$
by Condition {\bf B3}.}
Hence, {for a sufficiently small $\epsilon_0$ and any $x$},
\ba
 \sigma_1^{-1} g(x; G)
&\geq &
\frac{\alpha_1 }{\sigma_1\epsilon_0}
f_0\left ( \frac{x - \mu_1}{\epsilon_0} \right )
+
\frac{\alpha_2 }{\sigma_2}
f_0\left ( \frac{x - \mu_1}{\sigma_2} \right )
\nonumber \\
&\geq&
\frac{\alpha_1 }{\sigma_1}
f_0\left ( \frac{x - \mu_1}{\sigma_1} \right )
+
\frac{\alpha_2 }{\sigma_2}
f_0\left ( \frac{x - \mu_1}{\sigma_2} \right )
=
f(x; G).
\ea
This proves \eqref{lemma4.1}.

To prove \eqref{lemma4.2}, we first notice that
{the condition $\sup_x (1 + x^2) f_0(x) < \infty$ in Condition {\bf B3}  } implies
$(1 + x^2)^{2/3} f_0(x) < \delta_0$ when $|x| > M$ {for} some large $M$.
Let $\epsilon_0 < M^{-3/2}$.
When  $|x - \mu_1| \geq \sigma_1^{2/3}$,
we have
\[
|x - \mu_1|^2 / \sigma^2_1
\geq
\sigma_1^{-2/3}
\geq
\epsilon_0^{-2/3}
> M.
\]
Therefore, when  $|x - \mu_1| \geq \sigma_1^{2/3}$,
\be
\label{eqLemA2}
\frac{1 }{\sigma_1}
f_0\left ( \frac{x - \mu_1}{\sigma_1} \right )
=
\frac{\big [ 1 + \big ( \frac{x-\mu_1}{\sigma_1} \big )^2 \big ]^{3/2}
 f \big ( \frac{x-\mu_1}{\sigma_1} \big )}
{ \sigma_1 \big [ 1 + \big ( \frac{x-\mu_1}{\sigma_1} \big )^2 \big ]^{3/2}}
\leq
\frac{\delta_0 }
{ \sigma_1 \big [ 1 + \sigma_1^{-2/3} \big ]^{3/2}}
\leq
\delta_0.
\ee
Noting the additive term $\delta_0$ in $g(x; G)$, we find
\(
\log f(x; G) \leq \log g(x; G).
\)
{This completes the proof of (\ref{eqLemA2}). }
\end{proof}

\begin{lemma}
\label{lemA3}
Uniformly over $\sigma_1 < \epsilon_0$,
$\sigma_2 > \tau_0$ for some sufficiently small
$\epsilon_0$ and $\tau_0$, there exists
a sufficiently small $\delta_0$ such that
{\be
\label{lemma5.1}
\bbE \log \{ g(X; G) / f(X; G_0) \} < 0
\ee
}
{where $g(x; G)$ is defined in Lemma \ref{lemA2}
and $\bbE$ is taken under $f(x; G_0)$.}
\end{lemma}

When $\delta_0 = 0$, $g(x; G)$ is a density function.
Hence, the inequality holds by Jensen's inequality.
For each fixed $G$, $g(x; G)$ decreases when
$\delta_0 \downarrow 0$. Hence, there exists a $G$-specific
$\delta_0$ such that \eqref{lemma5.1} holds.
A unified $\delta_0$ is possible by going through
the finite open coverage property.

\begin{lemma}
\label{lemA4}
Let  $x_1, \ldots, x_n$ be {\iid} observations from $f(x; G_0)$.
Assume the conditions of Theorem \ref{thm2}.

(a) Uniformly on $G$ over $\sigma_1 \leq \sigma_2 < \epsilon_0$
for some sufficiently small $\epsilon_0$,
as $n \to \infty$, almost surely,
\[
\tilde \ell_n(G) <  \tilde \ell_n(G_0);
\]

(b)
Uniformly on $G$ over $\sigma_1 < \epsilon_0$ and $\sigma_2 > \tau_0$
for some sufficiently small $\epsilon_0$ and $\tau_0$,
as $n \to \infty$, almost surely,
\[
\sum_{i=1}^n \log g(x_i; G) < \sum_{i=1}^n \log f(x_i; G_0) = \ell_n(G_0),
\]
and therefore
\[
\tilde \ell_n(G) <  \tilde \ell_n(G_0);
\]

(c) The above two results imply that the maximum penalized
likelihood estimator of $G$ is consistent.
\end{lemma}

\begin{proof}
{We first consider (a)}.
By Lemma \ref{lemA1}, the number of observations within
a $\sigma_1^{2/3}$-neighborhood of any $\mu_1$ is no more
than
\(
n_1 =  2 n \sigma_1^{2/3} + 10 \log n.
\)
Similarly, this number for $\mu_2$ is no more than
\(
n_2 = 2 n \sigma_2^{2/3} + 10 \log n.
\)
Clearly, $n_1 + n_2 < n /2$ when both $\sigma_1, \sigma_2$
are smaller than $\epsilon_0$ and $n$ is large enough.

For any subset $A$ of $\{1, \ldots, n\}$, let
\[
\ell_n(G; A) = \sum_{i \in A} \log f(x_i; G).
\]
Let $A_h = \{i: |x_i - \mu_h| \leq  \sigma_h^{2/3}\}$ for $h=1, 2$.
Since $f(x; \btheta) \leq \sigma^{-1}$, we have
\[
\ell_n(G; A_1 \cup A_2)
\leq
{\blue -}n_1 \log (\sigma_1){\blue -} n_2 \log (\sigma_2).
\]
Taking the penalty $p_n(\sigma_1, \sigma_2)$ into consideration,
for any $\sigma_1, \sigma_2$ in the specified range,
we have
\be
\label{eqnA1}
\ell_n(G; A_1 \cup A_2) +p_n(\sigma_1, \sigma_2)
\leq
n \epsilon_1
\ee
for an arbitrarily small $\epsilon_1$
when $\epsilon_0$ is chosen sufficiently small.

At the same time, by inequality
\eqref{eqLemA2} in the proof of Lemma \ref{lemA2},
we find
\be
\label{eqnA2}
\ell_n(G; A^c_1 A^c_2)
\leq
\{n - n_1 - n_2\} \log \delta_0
\leq
(1/2) n \log \delta_0.
\ee
Combining \eqref{eqnA1} and \eqref{eqnA2}, we find
\[
\tilde \ell_n(G)
=
\ell_n(G; A_1 \cup A_2)  + \ell_n(G; A^c_1 A^c_2)  + p_n(\sigma_1, \sigma_2)+{p(\alpha_1,\alpha_2)}
\leq
n \{  \epsilon_1 +  (1/2) \log \delta_0 \}.
\]
At the same time,
\[
\tilde \ell_n(G_0)
= n \{ \bbE \log f(X; G_0)\} \{ 1 + o(1)\}.
\]
Clearly, when {$\delta_0$ is small enough}, we must have
\[
\tilde \ell_n(G) < \tilde \ell_n(G_0)
\]
almost surely. Since the inequality was obtained without considering a
specific $G$, the inequality holds uniformly
for all $G$. Hence, we have proved conclusion (a).

The first part of Conclusion (b) follows from the classical
consistency proof for the MLE by \cite{Wald1949}  and
the inequality (\ref{lemma5.1}) developed in Lemma \ref{lemA3}.
The difference between $\tilde \ell_n(G)$ and $\sum \log g(x_i; G)$
is bounded by $n_2 \log \sigma_2 + p_n(\sigma_2)$,
which is not large enough to change the direction of the inequality.
Hence, the second part of {the conclusion holds.}

Conclusions (a) and (b) imply that {the penalized
MLE} must be attained in the subspace of $G$ in which
$\epsilon_0 < \sigma_1; \tau_0 < \sigma_2$ almost surely.
The finite mixture model on this subspace  can be seen to
satisfy the conditions specified in the MLE
consistency proof of \cite{Kiefer1956}. Hence, the penalized
MLE is consistent.
\end{proof}

%
\vskip 14pt
\noindent {\large\bf Acknowledgements}

{
The work of Dr. Chen is supported by funding from Yunnan
University through One Thousand Talents,
and through funding from the Natural Science
and Engineering Research Council of Canada, RGPIN2014-03743.
Dr.~Li was supported in part by  the Natural Sciences and
Engineering Research Council of Canada, RGPIN-2015-06592.
The work of Dr.~Liu  is supported by 
a grant from the National Natural Science Foundation of China (11801359).}
\par



%

%
%
%
\end{document}

%% file: symbols.tex
\newcommand{\btheta}{\mbox{\boldmath $\theta$}}

\newcommand{\bm}{\mbox{\boldmath $m$}}

\newcommand{\bu}{\mbox{\bf u}}
\newcommand{\bv}{\mbox{\bf v}}
\newcommand{\bw}{\mbox{\bf w}}
\newcommand{\bt}{\mbox{\bf t}}
\newcommand{\bb}{\mbox{\bf b}}

\newcommand{\bB}{\mbox{\bf B}}

\newcommand{\bx}{\mbox{\bf x}}

\newcommand{\cR}{\mbox{$\mathcal{R}$}}

\newcommand{\bbE}{\mbox{$\mathbb{E}$}}

\newcommand{\ind}{\mbox{$\mathbbm{1}$}}

\newcommand{\GG}{\mathbb{G}}

\newcommand\raiseT[2]{%
  \setbox0\hbox{$#1{#2}$}\raise\dp0\box0}

\newcommand{\bea}{\begin{eqnarray}}
\newcommand{\eea}{\end{eqnarray}}
\newcommand{\ba}{\begin{eqnarray*}}
\newcommand{\ea}{\end{eqnarray*}}
\newcommand{\be}{\begin{equation}}
\newcommand{\ee}{\end{equation}}
\newcommand{\bi}{\begin{itemize}}
\newcommand{\ei}{\end{itemize}}

\newcommand{\var}{\mbox{\sc var}}

\newcommand{\cov}{\mbox{\sc cov}}

\newcommand{\iid}{\mbox{i.i.d.}}

\newcommand{\EM}{\mbox{\sc em}}

\newcommand{\cd}{\overset{d}{\longrightarrow}}

\newtheorem{theorem}{\bf Theorem}
\newtheorem{lemma}{\bf Lemma}

\newcommand{\vs}{\vspace{1ex}}
